\documentclass[sn-mathphys,Numbered]{sn-jnl}% Math and Physical Sciences Reference Style
\usepackage{mathtools}
\usepackage{bm}
\usepackage{bbold}

\usepackage{graphicx}%
\usepackage{multirow}%
\usepackage{amsmath,amssymb,amsfonts}%
\usepackage{amsthm}%
\usepackage{mathrsfs}%
\usepackage[title]{appendix}%
\usepackage{xcolor}%
\usepackage{textcomp}%
\usepackage{manyfoot}%
\usepackage{booktabs}%
\usepackage{algorithm}%
\usepackage{algorithmicx}%
\usepackage{algpseudocode}%
\usepackage{listings}%
\newtheorem{theorem}{Theorem}[section]
\newtheorem{corollary}{Corollary}[theorem]

\newtheorem{proposition}{Proposition}[section]
\theoremstyle{definition}
\newtheorem{definition}{Definition}[section]
\newtheorem{example}{Example}[section]
\theoremstyle{remark}
\newtheorem{remark}{Remark}[section]

\begin{document}

\title[The Mean Field Market Model Revisited]{The Mean Field Market Model Revisited}


\author*[1]{\fnm{Manuel} \sur{Hasenbichler}}\email{manuel.hasenbichler@tugraz.at}

\author[2]{\fnm{Wolfgang} \sur{Müller}}\email{w.mueller@tugraz.at}

\author[3]{\fnm{Stefan} \sur{Thonhauser}}\email{stefan.thonhauser@tugraz.at}

\affil[1,2,3]{\orgdiv{Institute of Statistics}, \orgname{Graz University of Technology}, \orgaddress{\street{Kopernikusgasse 24/III}, \city{Graz}, \postcode{8010}, \state{Styria}, \country{Austria}}}

% \affil[2]{\orgdiv{Department}, \orgname{Organization}, \orgaddress{\street{Street}, \city{City}, \postcode{10587}, \state{State}, \country{Country}}}

% \affil[3]{\orgdiv{Department}, \orgname{Organization}, \orgaddress{\street{Street}, \city{City}, \postcode{610101}, \state{State}, \country{Country}}}

\abstract{In this paper, we present an alternative perspective on the mean-field LIBOR market model introduced by Desmettre et al. in \cite{Desmettre}. Our novel approach embeds the mean-field model in a classical setup, but retains the crucial feature of controlling the term rate's variances over large time horizons. This maintains the market model's practicability, since calibrations and simulations can be carried out efficiently without nested simulations. In addition, we show that our framework can be directly  applied to model term rates derived from SOFR, ESTR or other nearly risk-free overnight short-term rates - a crucial feature since many IBOR rates are gradually being replaced. These results are complemented by a calibration study and some theoretical arguments which allow to estimate the probability of unrealistically high rates in the presented market models.}



\maketitle

\section*{Acknowledgments}
This research was funded in whole, or in part, by the Austrian Science Fund (FWF) P 33317. For the purpose of open access, the author has applied a CC BY public copyright licence to any Author Accepted Manuscript version arising from this submission. \bigskip

\section{Introduction and Motivation} \label{introduction}
In 2015, the Commission of the European Union amended the Solvency II\footnote{For further information regarding the Solvency II framework see {\url{https://www.eiopa.europa.eu/browse/regulation-and-policy/solvency-ii_en}} (Accessed: May 12, 2023).} Directive issued in 2009, which regulates all insurance businesses in the European Union (cf. \cite{2015Regulation}).
As a result, insurance companies operating in the European Union must assign market-consistent values to all balance sheet positions since the new regulation came into force in 2016. For the evaluation of liabilities - particularly in the context of life insurance contracts - this requires the market-consistent simulation of interest rates to determine the present value of future benefits, the so called \textit{best estimates}, which are agreed on in the corresponding insurance contracts (cf. \cite[Sec. 1]{GachHochgerner}). For the presentation and study of general cashflow models underlying insurance best estimates the interested reader is referred to \cite{GachHochgerner} and \cite{GachHochgerner2}.\\

Meanwhile, the Forward market model (or LIBOR market model), which models term rates for a finite and predetermined number of tenors, has gained increasing popularity over the last decades as it can be easily calibrated to market data. However, life insurance contracts can have maturities of up to 60 years, while liquid caplet and swaption prices, both of which are commonly used to calibrate the models, are typically only traded for maturities of twenty to thirty years in the future. For this reason, practitioners have to extrapolate from the model over a long period relative to the time horizon for which reliable market data is available. This practice, however, sometimes leads to the so-called phenomenon of \textit{term rate blow-ups}, especially if the model is calibrated in more turbulent market environments. In this piece of work, term rate blow-ups refer to the event where, at a certain future point in time, a significant proportion (e.g. more than $5$\%) of simulated term rates are unrealistically high (e.g. higher than $70$\%). Such issues, arising from both the general regulatory framework and from the application of this type of model in this context, have already been examined (e.g. cf. \cite{Verdani}, \cite{Desmettre}). In \cite{Desmettre}, Desmettre et al. propose to extend the classical Forward market model's dynamics to mean-field stochastic differential equations (MF-SDEs). This allows them to damp the coefficients driving the model's dynamics if the term rates' variances and hence the probability of blow-ups become too high. \\

To simulate forwards rates form their model, Desmettre et al. apply the particle approximation approach. However, this standard technique for the simulation of stochastic processes driven by MF-SDEs involves computationally expensive and time-consuming nested simulations and hence can be hardly applied for calibration purposes in practice. In Section \ref{sec:technical_results} of this work, we propose a new but equivalent "a posteriori" approach that enables practitioners to simulate from and calibrate the model in a time-efficient manner if the MF-SDE's coefficients depend only on time and a finite number of the  rates' moments as opposed to their entire distribution. In addition, we apply the ideas introduced in \cite{MercurioLyashenko} and \cite{MercurioLyashenko2} and formulate the results for "backward-looking term rates" that are e.g. stripped from overnight short-term rates such as SONIA, SOFR or ESTR too. In order to avoid arbitrariness (also cf. \cite[{Sec. 3.2.4}]{Verdani}), special emphasis is placed on the development of criteria to avoid overzealous reduction of the term rate’s variances that is inconsistent with market observations. In the appendix, we present a technique which allows us to a-priori estimate the effectiveness of a chosen dampening approach. \\

We start by introducing the stochastic basis that is adapted from \cite{Desmettre}. Let ${0 = T_0 < T_1 < \dots < T_N < \infty}$, $N \in \mathbb{N}$ be a prespecified set of settlement dates, $\left(\Omega,\mathcal{A},\mathcal{F},\mathbb{Q}^N\right)$ be a filtered probability space supporting Brownian motion and suppose that the filtration $\mathcal{F}$ fulfills the usual conditions\footnote{i.e. $\mathcal{F}$ is right-continuous and $\mathbb{Q}^N$-complete.} with $\mathcal{A} = \mathcal{F}_{T_N}$. Moreover, denote $\mathbb{Q}^i$ the $T_i$-forward measure for all $1\leq i\leq N$, $\mathcal{P}\left(\mathbb{R}\right)$ the set of probability measures on $\mathbb{R}$ and define \[\mathcal{P}_k\left(\mathbb{R}\right) = \left\{\mu \in \mathcal{P}\left(\mathbb{R}\right): \int_\mathbb{R} x^k \, \mathrm{d}\mu(x) < \infty\right\}, \quad k \in \mathbb{N}.\]
As proposed in \cite{Desmettre}, we model the term rates ${\{F_i\}_{i=1}^N := \{t \mapsto F(t,T_{i-1},T_i)\}_{i=1}^N}$ associated with this tenor structure by  
\begin{equation} \label{equ:fwd_dynamics}
\begin{cases}
\mathrm{d}F_i(t) = F_i(t)\sigma_i\left(t, \mu_t^i\right)^T \mathrm{d}W_t^i \quad 0 \leq t \leq T_{i-1}, \\
F_i(0) = F\left(0,T_{i-1},T_i\right) > 0, 
\end{cases}
\end{equation}
if the term rates are forward-looking (e.g. IBORs such as EURIBOR or LIBOR) and
\begin{equation} \label{equ:bwd_dynamics}
\begin{cases}
\mathrm{d}F^B_i(t) = F^B_i(t)\sigma^B_i\left(t, \mu_t^i\right)^T \mathrm{d}W_t^i \quad 0 \leq t \leq T_i, \\
F^B_i(0) = F^B\left(0,T_{i-1},T_i\right) > 0, 
\end{cases}
\end{equation}
if the term rates are backward-looking (e.g. term rates stripped from overnight short-term rates such as SONIA, SOFR or ESTR). Hereby, $W^i$ is $\mathbb{Q}^i$ $d$-dimensional Brownian motion, $\mu_t^i \in \mathcal{P}_k\left(\mathbb{R}\right)$ is the law of $F_i(t)$ respectively $F^B_i(t)$ under $\mathbb{Q}^i$ and \[\sigma_i: [0,T_{i-1}]\times\mathcal{P}_k\left(\mathbb{R}\right)\to\mathbb{R}^d \] in the first case and 
\[\sigma^B_i: [0,T_i]\times\mathcal{P}_k\left(\mathbb{R}\right)\to\mathbb{R}^d\]
in the latter case. In this paper, we mostly consider $d = N$ for the purpose of simplicity. If we denote by $R(T_{i-1},T_i) = F_i(T_{i-1})$ the forward-looking interest that is accrued in $[T_{i-1},T_i]$ and by $R^B(T_{i-1},T_i) = F^B_i(T_i)$ the backward-looking interest that is accrued in $[T_{i-1},T_i]$, backward and forward-looking interest rates conceptionally only differ in their measurability in lognormal market models. In fact, Lyashenko and Mercurio show the following proposition in \cite{MercurioLyashenko}: \\

\begin{proposition}{\cite[Sec. 2.4]{MercurioLyashenko}} \label{prop:equivalence_fwd_bwd_rates} ~\\
Let $0 < S < T$. Then, $F^B(t,S,T) = F(t,S,T) \quad \forall \, t \in [0,S]$. \bigskip
\end{proposition}

Thus, we refrain from stating analogous results and omit the superscript $B$ in the following.\\

As in the classical theory on the Forward market model, the various forward measures are linked by the change of measure processes
\[ \frac{\mathrm{d}\mathbb{Q}^{i-1}}{\mathrm{d}\mathbb{Q}^{i}} = \exp\left(\int_0^\cdot \frac{\Delta_i F_i(s)}{1 + \Delta_i F_i(s)} \left(\sigma_i\left(s,\mu_s^i\right)\right)^T\mathrm{d}W_s^i - \frac{1}{2}\int_0^\cdot \left|\frac{\Delta_i F_i(s)}{1 + \Delta_i F_i(s)} \sigma_i\left(s,\mu_s^i\right)\right|^2\mathrm{d}s \right),\]
assuming that the term rates $\{F_i\}_{i=1}^N$ exist on $\left(\Omega,\mathcal{A},\mathcal{F}\right)$. We recall that this ensures that there exists a probability measure $\mathbb{Q}^* \sim \mathbb{Q}^N$ such that the price at time $T_k$ of any $\mathcal{F}_{T_N}$-measurable and attainable claim $X$ is \[\pi_X(T_k) = B^*(T_k) \, \mathbb{E}_{\mathbb{Q}^*}\left[\frac{X}{B^*(T_N)}\middle|\mathcal{F}_{T_k}\right] \quad \forall \, 1 \leq k \leq N\] where \[ \begin{cases}
B^*(0) = 1, \\
B^*(T_i) = \left(1 + \Delta_i F_i(T_i)\right))B^*(T_{i-1}), \quad \forall \, 1 \leq i \leq N
\end{cases} \] is the discrete-time implied bank account process (e.g. c.f. \cite[Chpt. 11]{Filipovic}).
In the case of backward-looking term rates, this bank account process is no-longer previsible but still monotonically increasing and adapted with respect to $\left\{\mathcal{F}_{T_i}\right\}_{i=1}^N$. \bigskip

\section{Theoretical Considerations} \label{sec:technical_results}
If the instantaneous volatilities $\sigma_i\left(t,\mu_t^i\right)$ are only functions of time and a finite number of moments of the $i$-th term rate and do not depend on the entire distribution $\mu^i$ ${\forall \, 1 \leq i \leq N}$, damping effects can already be achieved in the classical Forward market model framework, as Theorem \ref{bwd_special_existence} below suggests. In particular, it proposes a method that allows practitioners to achieve the same damping effects as presented by \cite{Desmettre} without performing time-consuming nested simulations. Nonetheless, the ordinary differential equations (ODEs) arising from this method need to be solved numerically. Therefore, in the course of this section, we additionally propose a slightly different but equivalent damping approach for which the solutions of these ordinary differential equations can usually be computed explicitly. This allows the easy-to-implement simulation of forward- and backward-looking term rates.

\subsection{Existence and Uniqueness}
We show that the existence of the term rates $\left\{F_i\right\}_{i=1}^N$ in our special model is intertwined with the existence of solutions to a set of ODEs. \\

\begin{theorem} \label{bwd_special_existence}~\\
Suppose that for all $1 \leq i \leq N$ the function $\ \displaystyle{\sigma_i: [0,T_i] \times \mathcal{P}_k\left(\mathbb{R}\right) \to \mathbb{R}^d}$ has the form \[\sigma_i\left(t,\mu_t^i\right) = \lambda_i\left(t,\mathbb{E}_{\mathbb{Q}^i}\left[\left(F_i(t)\right)^2\right],\dots,\mathbb{E}_{\mathbb{Q}^i}\left[\left(F_i(t)\right)^k\right]\right)\] where ${\lambda_i:  [0,T_i] \times \left(\mathbb{R}^+\right)^{k-1}} \to \mathbb{R}^d$ is measurable and its length $\left|\lambda_i\right|$ is piecewisely continuous in every coordinate and non-zero. Then, the following statements hold:
\begin{enumerate}
\item Let $\left(F_1,\dots,F_N\right)$ be a solution of (\ref{equ:bwd_dynamics}) on $\left(\Omega,\mathcal{F},\mathbb{Q}^N\right)$ and suppose $t \mapsto \left|\sigma_i(t,\mu_t^i)\right|$ is bounded on $[0,T_i]$. Then, for every $1 \leq j \leq N$ the function $\displaystyle \psi_i^j(t) := \mathbb{E}_{\mathbb{Q}^i}\left[F_i(t)^j\right]$ is the unique continuous as well as piecewisely differentiable solution of
\begin{equation} \label{ODE:psi_gen_2}
\begin{cases}
\frac{\mathrm{d}}{\mathrm{d}t}\psi_i^j(t) = \frac{j(j-1)}{2} \psi_i^j(t) \left| \lambda_i(t,\psi_i^2(t),\dots,\psi_i^k(t))\right|^2 \quad \left(t \in \mathcal{D}_i\cap(0,T_i)\right), \\ 
\psi_i^j(0) = \left(F_i(0)\right)^j,
\end{cases}
\end{equation}
where $\displaystyle \mathcal{D}_i := \left\{t \, : \, \psi_i^j \text{ is differentiable in } t \text{ for all } 2 \leq j \leq k\right\}$ for all $1 \leq i \leq N$.
\item Suppose there exist continuous and piecewisely differentiable $\left(\psi_i^j\right)_{j=2}^k$ for all ${1 \leq i \leq N}$ such that for all $i,j$ the function ${\psi_i^j: [0,T_i] \to \mathbb{R}^+}$ satisfies 
\begin{equation} \label{ODE:psi_gen_1}
\begin{cases}
\frac{\mathrm{d}}{\mathrm{d}t}\psi_i^j(t) = \frac{j(j-1)}{2} \psi_i^j(t) \left| \lambda_i(t,\psi_i^2(t),\dots,\psi_i^k(t))\right|^2 \quad \left(t \in \mathcal{D}_i\cap(0,T_i)\right), \\ 
\psi_i^j(0) = \left(F_i(0)\right)^j,
\end{cases}
\end{equation}
where $\displaystyle \mathcal{D}_i := \left\{t \, : \, \psi_i^j \text{ is differentiable in } t \text{ for all } 2 \leq j \leq k\right\}$.\\
Then, for all $1 \leq i \leq N$ there exists an unique strong solution to (\ref{equ:bwd_dynamics}) on $\left(\Omega,\mathcal{F},\mathbb{Q}^N\right)$ and for all $i,j$ \begin{equation}\label{eq:psi} \mathbb{E}_{\mathbb{Q}^i}\left[F_i(t)^j\right] = \psi_i^j(t) \quad 0 \leq t \leq T_i.\end{equation}
\end{enumerate}
\end{theorem}

\begin{proof}~
\begin{enumerate}
\item By construction,
\[ F_i(t)^j = F_i(0)^j \, \exp\left(j \, \int_0^t\sigma_i(s,\mu_s^i) \mathrm{d}W_s^i - \frac{j}{2} \int_0^t\sigma_i(s,\mu_s^i)^2 \mathrm{d}s\right)\]
under $\mathbb{Q}^i$, where \[\sigma_i\left(t,\mu_t^i\right) = \lambda_i\left(t,\mathbb{E}_{\mathbb{Q}^i}\left[F_i(t)^2\right],\dots,\mathbb{E}_{\mathbb{Q}^i}\left[F_i(t)^k\right]\right).\] 
For $2 \leq j \leq k$ set $\psi_i^j(t) := \mathbb{E}_{\mathbb{Q}^i}\left[F_i(t)^j\right]$. Then,
\begin{equation} \label{equ:psi_thm}
\psi_i^j(t) = F_i(0)^j \exp\left(\frac{j(j-1)}{2}\int_0^t \left|\lambda_i\left(s,\psi_i^j(s),\dots,\psi_i^k(s)\right)\right|^2 \mathrm{d}s\right).
\end{equation}
Since $t \mapsto \sigma_i(t,\mu_t^i)$ is bounded on $[0,T_{i}]$, $\psi_i^j$ are continuous on $\left[0,T_{i}\right]$. In addition, $\left(\psi_i^j\right)_{j=2}^k$ are piecewisely differentiable since $\left|\lambda_i\right|$ is piecewisely continuous in every coordinate. Differentiating the above expression w.r.t. $t \in \mathcal{D}_i$ yields (\ref{ODE:psi_gen_2}). \\

\item Define $\displaystyle \tilde{\sigma}_i(t) := \lambda_i\left(t,\psi_i^2(t),\dots,\psi_i^k(t)\right)$ for all $1 \leq i \leq N$ where $\left(\psi_i^j\right)_{j=2}^k$ are the continuous and piecewisely differentiable functions that solve (\ref{ODE:psi_gen_1}). It is well-known (e.g. c.f. \cite{MusielaRutkowski}, Ch. 12.4) that there exist unique strong solutions of 
\[ \begin{cases}
\mathrm{d}F_i(t) = F_i(t)\tilde{\sigma}_i(t)^T \mathrm{d}W_t^i \quad 0 \leq t \leq T_i, \\
F_i(0) = F_i\left(0,T_{i-1},T_i\right) > 0, 
\end{cases} \]
on $\left(\Omega,\mathcal{F},\mathbb{Q}^N\right)$ since the $\tilde{\sigma}_i$ are bounded and piecewisely continuous on $[0,T_{i}]$.
It remains to show that $\psi_i^j(t) = \mathbb{E}_{\mathbb{Q}^i}\left[F_i(t)^j\right]$ on $[0,T_{i}]$ for all $i$ and $j$. \\

By construction,
\[ F_i(t)^j = F_i(0)^j \, \exp\left(j \, \int_0^t \tilde{\sigma}_i(s) \mathrm{d}W_s^i - \frac{j}{2} \int_0^t\tilde{\sigma}_i(s)^2 \mathrm{d}s\right)\]
under $\mathbb{Q}^i$ and therefore \begin{equation}\label{proof:moments_2} \mathbb{E}_{\mathbb{Q}^i}\left[F_i(t)^j\right] = F_i(0)^j \exp\left(\frac{j(j-1)}{2}\int_0^t \left|\tilde{\sigma}_i(s)\right|^2 \mathrm{d}s\right).\end{equation}
Since the right hand side of this equation is continuous and piecewisely differentiable in $t$, so is $\mathbb{E}_{\mathbb{Q}^i}\left[F_i(t)^j\right]$. What is more, \[\mathbb{E}_{\mathbb{Q}^i}\left[F_i(0)^j\right] = F_i(0)^j = \psi_i^j(0)\] and note that $\mathbb{E}_{\mathbb{Q}^i}\left[F_i(t)^j\right]$ is differentiable in $t$ iff $\left|\lambda_i\left(t,\psi_i^2(t),\dots,\psi_i^k(t)\right)\right|^2$ is continuous in $t$. Hence, $\displaystyle \mathcal{D}_i = \left\{t \, : \,  \mathbb{E}_{\mathbb{Q}^i}\left[F_i(t)^j\right] \text{ is differentiable in } t\right\}.$ By differentiating (\ref{proof:moments_2}) in $\in \mathcal{D}_i$, one finds that $\mathbb{E}_{\mathbb{Q}^i}\left[F_i(0)^j\right]$ is the unique solution of (\ref{ODE:psi_gen_1}). \qedhere
\end{enumerate} 
\end{proof}

\begin{remark}
    Equation (\ref{equ:psi_thm}) shows that the damping of higher moments ($k > 2$) is equivalent to the damping of the second moments of the forward rates.
\end{remark} \bigskip

We primarily strive to present a method to mitigate interest rate blow-ups when employing the Forward market model. By the last remark, it is sufficient to consider the case $k=2$ of Theorem \ref{bwd_special_existence}. Similar to the approach developed in \cite{Desmettre}, we use the term rates' total variance $\phi_i$ instead of $\psi_i^2$ as a control variable to reduce blow-ups. $\phi_i$ classically arises in option trading and corresponds to $T_i \cdot\left(\sigma_i^{BS}\right)^2$, where $\sigma_i^{BS}$ is the Black implied $F_i$-caplet volatility. In our log-normal model setting, 
\begin{equation} \label{equ:phi}
\phi_i(t) = \log\left(\frac{\psi_i^2(t)}{F_i(0)^2}\right) = \log\left(\frac{\mathbb{E}_{\mathbb{Q}^i}\left[F_i(t)^2\right]}{F_i(0)^2}\right) = \int_0^t \left|\sigma_i(s,\mu_s^i)\right|\mathrm{d}s.
\end{equation}
This is just a reparametrisation. Note that in terms of $\phi_i$, Equation (\ref{ODE:psi_gen_1}) is equivalent to 
\begin{equation} \label{equ:problem_reparametrised}
   \begin{cases}
        \frac{\mathrm{d}}{\mathrm{d}t}\phi_i(t) = \frac{1}{\psi_i^2(t)}\frac{\mathrm{d}}{\mathrm{d}t}\psi_i^2(t) = \lvert\sigma_i(t,\mu_t^i)\rvert^2 \quad \left(t \in \mathcal{D}_i\cap(0,T_i)\right), \\ 
        \phi_i(0) = 0.
    \end{cases} 
\end{equation}
The following Corollary provides a practically useful choice of $\sigma_i(t,\mu_t^i)$, which allows to so solve Equation (\ref{equ:problem_reparametrised}) by separating variables. \\

\begin{corollary} \label{bwd_corollary_special_existence}~\\
Suppose that for all ${1 \leq i \leq N}$ the function $\displaystyle \sigma_i: [0,T_i] \times \mathcal{P}_2\left(\mathbb{R}\right) \to \left(\mathbb{R}^+\right)^d$ has the form
\[\sigma_i\left(t,\mu_t^i\right) = g_i(t) f_i\left(\phi_i(t)\right) v_i\left(t,\phi_i(t)\right)\] 
with $\phi_i$ as in (\ref{equ:phi}), where $g_i: [0,T_i] \to \mathbb{R}^+$, $f_i:  \mathbb{R}_0^+ \to \mathbb{R}^+$ are left-continuous as well as piecewisely continuous and $v_i: [0, T_i] \times \mathbb{R}_0^+ \to \mathcal{S}^{d-1} = \left\{x \in \mathbb{R}^d \, : \, \lvert x \rvert_2 = 1 \right\}$ is left-continuous. 
In addition, assume that \begin{equation} \label{equ:inequ_cor} \int_0^{T_{i}} g_i(s)^2 \mathrm{d}s \leq \int_0^\infty \frac{1}{f_i(z)^2} \mathrm{d}z \quad \forall \, 1 \leq i \leq N.\end{equation} 
Then, $F_i$ exists as an unique strong solution to (\ref{equ:bwd_dynamics}). Furthermore, 
\begin{equation} \label{bwd_explicit_variance}
\phi_i(t) = V_i\left(\int_0^t g_i(s)^2 \mathrm{d}s\right)
\end{equation}
where \[V_i^{(-1)}(x) = \int_0^x f_i(z)^{-2} \mathrm{d}z.\]
\end{corollary}

\begin{proof}
By assumption, $f_i>0$ and so, $x \mapsto \int_0^x f_i(z)^{-2} \mathrm{d}z$ is well-defined, continuous and strictly monotonically increasing. Thus, its inverse ${V_i: [0,y_i) \to [0,\infty)}$ with ${y_i := \int_0^\infty f_i(z)^{-2} \mathrm{d}z \leq \infty}$ exists and, given (\ref{equ:inequ_cor}), is well-defined on $\left[0,\int_0^{T_i} g_i(s)^2 \mathrm{d}s\right]$. Separation of variables shows that (\ref{bwd_explicit_variance}) solves (\ref{equ:problem_reparametrised}). Then, Theorem (\ref{bwd_special_existence}) yields the claim. \qedhere
\end{proof}

\begin{remark} \label{remark_factors}
While the $\left\{g_i\right\}_{i=1}^N$ govern the principle volatility structure over time and are as such calibrated to market data as will be shown in Section \ref{subsection:calibration}, the $\left\{f_i\right\}_{i=1}^N$ are responsible for damping whenever the control variables $\left\{\phi_i\right\}_{i=1}^N$ become too large. The $\left\{v_i\right\}_{i=1}^N$, on the other hand, determine the term rates' correlations. Thus, the $\left\{g_i\right\}_{i=1}^N$ are referred to as \textit{principle factors}, the $\left\{v_i\right\}_{i=1}^N$ are called \textit{correlation factors} and the $\left\{f_i\right\}_{i=1}^N$ are referred to as \textit{damping factors} of the instantaneous volatilities $\left\{\sigma_i(t,\mu_t^i)\right\}_{i=1}^N$ respectively. Note that $f_i \equiv 1$ ${\forall \, 1 \leq i \leq N}$ yields the classical Forward market model framework.
\end{remark}

\subsection{Total Implied Variance Structures} \label{subsection:total_implied_volatility_structures}
We need to choose $\left\{f_i\right\}_{i=1}^N$ in Corollary \ref{bwd_corollary_special_existence} in a way such that we can control the damping of the total implied variances $\left\{\phi_i\right\}_{i=1}^N$. We achieve this by specifying a suitable class of functions for $\left\{V_i\right\}_{i=1}^N$ in (\ref{bwd_explicit_variance}) which allows the easy and explicit computation of $\mathbb{E}_{\mathbb{Q}^i}\left[F_i(t)^2\right]$. To this end, we start with three definitions. \\

\begin{definition}
$k: \mathbb{R}_0^+ \to \mathbb{R}_0^+$ is a \textit{(total implied variance structure) segment} if
\begin{enumerate}
\item $k$ is strictly monotonically increasing and concave or convex,
\item $k$ is continuously differentiable on $\mathbb{R}^+$,
\item $k(0) = 0$.
\end{enumerate}
\end{definition}

\begin{definition} \label{def:total_implied_variance_structure}
A bijective map $V: \mathbb{R}^+_0 \to \mathbb{R}^+_0$ is \textit{induced by segments} $\left(k_i\right)_{i=0}^M$ for some $M \in \mathbb{N}_0$ if there exist strictly monotonically increasing \textit{(total implied variance) thresholds} ${0 = \tau_0 < \dots < \tau_{M+1}=\infty}$ such that
\[ V(y) = \sum_{i=0}^{M} \mathbb{1}_{[\tau_i,\infty)}(y) \, k_i(y \wedge \tau_{i+1} - \tau_i). \] 
We call such a function $V$ \textit{total implied variance structure}. Note that $V$ is continuous and piecewisely differentiable on ${\mathcal{D}_V := \mathbb{R}_0^+\backslash\{\tau_i\}_{i=0}^{M}}$.
\end{definition} \bigskip

\begin{definition} \label{def:damping_factor}
A left-continuous function $f: \mathbb{R}_0^+ \to \mathbb{R}^+$ is called a \textit{(damping) factor} if there exist $0 = \eta_0 < \dots < \eta_M < \infty$ for some ${M \in \mathbb{N}_0}$ if $f$ is continuous and either monotonically increasing or monotonically decreasing on each $(\eta_i,\eta_{i+1}]$ for ${0 \leq i \leq M-1}$ and on $(\eta_M,\infty)$.
\end{definition} \bigskip

\begin{figure}[h]
\centering
    \begin{minipage}[t]{0.48\textwidth}
        \includegraphics[width=1\textwidth]{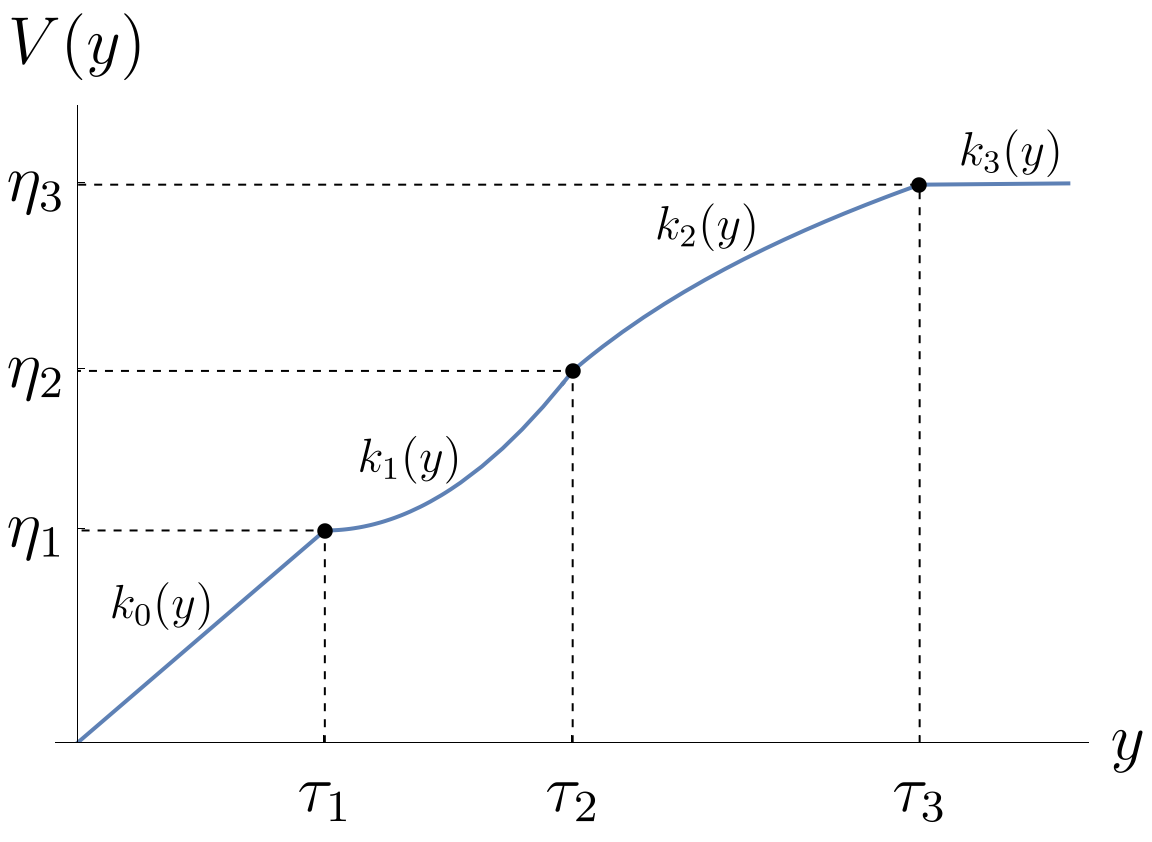}
        \caption{Total implied volatility structure, exemplary plot}
    \end{minipage} \hspace{0.3cm}
    \begin{minipage}[t]{0.48\textwidth}
        \includegraphics[width=1\textwidth]{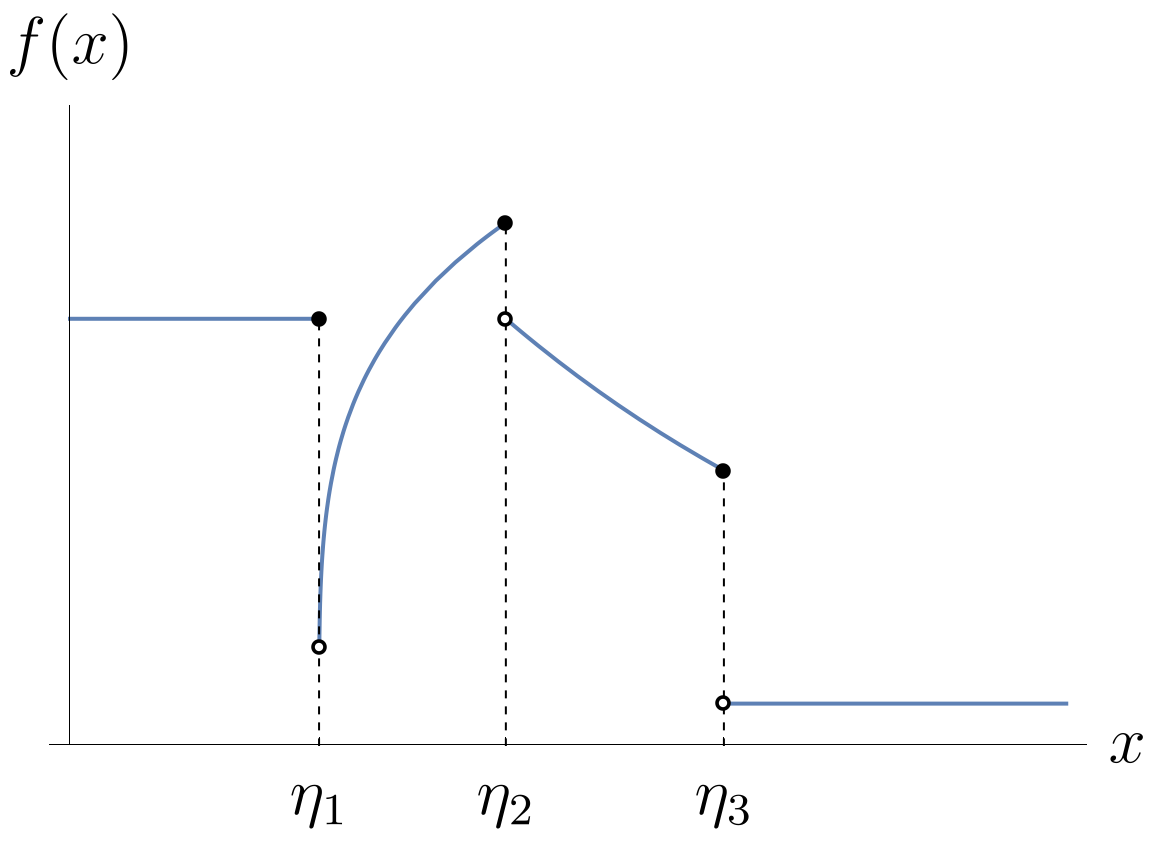}
        \caption{Damping factor, exemplary plot}
    \end{minipage}
\end{figure}

As a direct consequence of Definition \ref{def:total_implied_variance_structure} and Definition \ref{def:damping_factor}, there exists a damping factor $f$ for every total implied variance structure $V$ such that $V^{(-1)} = \int_0^\cdot f(z)^{-2} \mathrm{d}z$ and vice-versa. Indeed, given a total implied variance structure $V$ as in Definition \ref{def:total_implied_variance_structure}, one such damping factor $f$ is 
\begin{align} \label{equ:assoc_damping_factor}
    f(x) = \mathbb{1}_{\{0\}}(x) f_0 &+ \sum_{i=0}^{M-1} \mathbb{1}_{(\eta_i,\eta_{i+1}]}(x)\left(\left[\frac{\mathrm{d}}{\mathrm{d}x}k_i^{(-1)}\right]\left(x-\eta_i\right)\right)^{-\frac{1}{2}} \nonumber \\
    &+ \mathbb{1}_{(\eta_M,\infty)}(x)\left(\left[\frac{\mathrm{d}}{\mathrm{d}x}k_M^{(-1)}\right]\left(x-\eta_M\right)\right)^{-\frac{1}{2}}
\end{align} 
with $\eta_i := V(\tau_i)$ $\forall 1 \leq i \leq M$ and some $f_0 \in \mathbb{R}^+$. Usually, $f_0 := f(0+)$ and we will refer to this particular choice as \textit{the associated} damping factor. \\

In this paper, the focus is on total implied variance structures that damp the term rates' instantaneous volatilities and consist of at most three different concave segments. Furthermore, we choose the same total implied volatility structure for each term rate to extrapolate the given market data. In particular, we always set $k_0 = \text{id}_{\mathbb{R}_0^+}$ and $\tau_1 > 0$ for reasons that we explain in more detail in Subsection \ref{subsubsection:damping_thresholds}. \\
However, we note that parameterised segments such as $\left\{y \mapsto a_j y^{p_j}\right\}_{j = 1}^L$ with $a_j, p_j > 0$ for every $1 \leq j \leq L$ can be applied to fit any strictly increasing total implied volatility structure without the need to individually adjust the term rates' instantaneous volatilities by multiplicative constants. This is especially important when extrapolating market data. In this way, our framework also extends the literature on market model calibration. \\
Two examples of total implied variance structures and their associated damping factors are given below. Of course, these basic structures can be combined to more complex ones. \\

\begin{remark}
Corollary \ref{bwd_corollary_special_existence} shows that we have to consider both, $V$ and $V^{(-1)}$, in the computation of $\left\{\phi_i\right\}_{i=1}^N$.
As a result, easily invertible segments $k_0, \dots, k_M$ will be chosen in practice.
\end{remark} \bigskip

\begin{example}[adapted from \cite{Desmettre} by Desmettre et al.] \label{natural_damping} ~\\
Let $\tau_1 = \tau > 0$ and $k_1(y) = k(y) =  \frac{\tau}{2} \log\left(1 + \frac{2 y}{\tau}\right)$. This is a natural choice as $k$ is strictly monotonically increasing, concave, continuously differentiable and satisfies $k(0) = 0$ as well as ${k'(0) = 1}$. Thus, $k$ is a segment that, together with the threshold $\tau$, induces the continuously differentiable total implied variance structure $V: \mathbb{R}_0^+ \to \mathbb{R}_0^+$ given by \[ V(y) = \begin{cases}
y & \text{if }  y \leq \tau \\
\tau + k(y-\tau) & \text{if } \tau \leq y
\end{cases} = \begin{cases}
y & \text{if }  y \leq \tau, \\
\tau + \frac{\tau}{2} \, \log\left(\frac{2 y}{\tau} - 1\right) & \text{if } y > \tau.
\end{cases} \]
By (\ref{equ:assoc_damping_factor}), its associated damping factor is 
\[ f(x) = \mathbb{1}_{[0,\tau]}(x) + \mathbb{1}_{(\tau,\infty)}(x)\exp\left(-\frac{x-\tau}{\tau}\right) = \exp\left(-\frac{\left(x-\tau\right)^+}{\tau}\right).\]
\end{example} \bigskip

\begin{example}[pseudo-volatility freeze] \label{vol_freeze}
Let $\tau > 0$, $\epsilon, \delta \in [0,1)$ and \[f(x) = \begin{cases} 
1 & \text{if } x \leq \tau(1-\delta), \\
1 - (x - (1 - \delta)\tau)\frac{1 - \epsilon}{\delta\tau} & \text{if } \tau(1-\delta) < x \leq \tau, \\
\epsilon & \text{if } \tau < x.
\end{cases}\]
Again, $f$ is a damping factor and \[ H(x) = \int_0^x f(z)^{-2} \mathrm{d}z = \begin{cases}
x & \text{if } x \leq \tau(1-\delta), \\
(1 - \delta)\tau + \frac{\delta \tau (x - (1 - \delta) \ \tau)}{\tau - x (1 - \epsilon) - (1 - \delta) \epsilon \tau} & \text{if } \tau(1-\delta) < x \leq \tau, \\
\frac{x - (1 - \epsilon (\delta + \epsilon - \delta \epsilon)) \tau}{\epsilon^2} & \text{if } \tau < x.
\end{cases}\]
Its inverse is \[ V(y) = \begin{cases}
y & \text{if } y \leq \tau(1-\delta), \\
-\frac{\tau (y - y (1 - \delta) \epsilon - (1 - \delta)^2 (1 - \epsilon) \tau)}{\tau - y (1 - \epsilon) - \delta (2 - \epsilon) \tau - \epsilon \tau} & \text{if } \tau(1-\delta) < y \leq \tau + \delta \left(\frac{1}{\epsilon} - 1\right) \tau, \\
\tau + \epsilon (y \epsilon - \delta (1 - \epsilon) \tau - \epsilon \tau) & \text{if } \tau + \delta \left(\frac{1}{\epsilon} - 1\right) \tau < y.
\end{cases} \]
Here, $\tau_1 = (1-\delta)\tau$, $\tau_2 = \tau + \delta \left(\frac{1}{\epsilon} - 1\right) \tau$ are the total implied variance structure thresholds and 
\[k_1(y) = \frac{y \delta \tau}{y(1-\epsilon) + \delta \tau}, \quad k_2(y) = y \epsilon^2\] 
are the corresponding segments. This damping method is referred to as pseudo-volatility freeze since it mimics a volatility freeze for $\epsilon << 1$ and $\delta << 1$. However, unlike the industry practice of volatility freezing, the model can remain market consistent, although radically damped if the term rates' total implied variances exceed predefined thresholds.
\end{example} \bigskip

In order to mitigate the blow-ups of term rates, this section has put the focus on the damping factors so far. However, \cite{Desmettre} observed in their simulation study that the probability of simulating unrealistically high term rates under the spot measure $\mathbb{Q}^*$ can also be greatly reduced by changing the term rates' instantaneous correlation structure. Here, the maps $\left\{v_i\right\}_{i=1}^N$ that model the term rates' instantaneous correlations (see also Remark \ref{remark_factors}) provide the theoretical framework for this idea. \\

\begin{example}[decorrelation beyond a threshold] \label{decorrelation}
This choice of correlation factors is based on \cite{Desmettre}. Let $1 \leq i \leq N$ and consider a model with $d=N$-dimensional Brownian motion in the background. Moreover, let $1 \leq d_1 \leq d$ and $\left\{u_i\right\}_{i=1}^{d_1}$ be the unit vectors in $\mathbb{R}^{d_1}$. The role of the later will be discussed in the next Subsection. We call the choice of correlation factors of the form
\[ v_i(x) = \begin{cases}
\iota_N\left(u_i\right) & \text{if } x \leq \tau_1, \\
e_i &  \text{if } x > \tau_1,
\end{cases}\] \textit{decorrelation beyond a threshold} where $\iota_N: \mathbb{R}^{d_1} \to \mathbb{R}^d$ denotes the canonical embedding of $\mathbb{R}^{d_1}$ in $\mathbb{R}^d$ and $\{e_i\}_{i=1}^d$ is the canonical orthonormal basis of $\mathbb{R}^d$. Observe that the $i$-th term rate's dynamics in this example is driven by standard $d$-dimensional Brownian motion, although only standard $d_1$-dimensional Brownian motion is required if its total implied variance $x$ does not exceed the threshold $\tau_1$.
\end{example}

\subsubsection{Damping Thresholds and Market Consistency} \label{subsubsection:damping_thresholds}
In this subsection, we discuss an important aspect of the model parametrisation, i.e. the choice of damping thresholds to ensure "market consistency" when damping the term rates' instantaneous volatilities. Hereby, a market model is generally referred to as \textit{market consistent}, if the model (approximately) yields the same prices as the market for the interest rate derivatives to which it has been calibrated. Of course, this is not a rigorous mathematical definition. Mathematically, we understand this important model feature in the following way: \\

\begin{definition}
Let $\left\{L_i\right\}_{i=1}^N$ be undamped term rates that satisfy 
\begin{equation*}
\begin{cases}
\mathrm{d}L_i(t) = L_i(t) g_i(t) \left(u_i\right)^T \mathrm{d}W_t^i \quad 0 \leq t \leq T_i, \\
L_i(0) = F\left(0,T_{i-1},T_i\right) > 0,
\end{cases}
\end{equation*}
under $\mathbb{Q}^i$ for all $1 \leq i \leq N$. Moreover, let $\left\{F_i\right\}_{i=1}^N$ be the damped term rates that solve (\ref{equ:bwd_dynamics}). Then, these two models are \textit{consistent} for the tenors ${T_1 < \dots < T_K}$ (${1 \leq K \leq N}$) if the model prices of all caplets and swaptions defined on this subtenorstructure coincide.
\end{definition} \bigskip

\begin{remark}
Clearly, the two models induce different cap and swaption prices for caps and swaptions with tenors $T_{K+1} < \dots < T_N$.
\end{remark} \bigskip

As mentioned before, we choose the same total implied variance structure, $V$, for all term rates. If we assume that ${k_0 = \mathrm{id}_{\mathbb{R}^+_0}}$ and that the correlation structure of the damped model is chosen as described at the end of Subsection \ref{subsection:total_implied_volatility_structures}, we immediately conclude that the two models above are consistent iff 
\begin{equation} \label{consistency_caps}
    \tau_1 \geq \tau_{\min} := \max_{1 \leq i \leq K} \int_0^{T_i} g_i(s)^2 \mathrm{d}s
\end{equation} and 
\begin{equation} \label{consistency_swaptions}
    v_i(x) = u_i
\end{equation}
for any $x \in [0,\tau_1]$ and for all $1 \leq i \leq K$. In order to observe a damping effect, note that $\tau_1 >  0$ must satisfy \[{\tau_1 < \tau_{\max} := \int_0^{T_{N}} g(s)^2 \mathrm{d}s}.\] Although the idea presented in this Subsection sets lower and upper bounds for the first damping threshold, it does not provide any insight into what would characterise a reasonable choice of damping factors and damping thresholds. One approach to solve this problem is presented in the Appendix. \\

From a calibration perspective, this allows us to separate calibration and damping and calibrate the classical Forward market model to the given market data at first. If this model sufficiently captures the market dynamics for the tenors $T_1 < \dots < T_K$ for which cap and swaption prices are observed, so does the damped model if it satisfies (\ref{consistency_caps}) and (\ref{consistency_swaptions}). We say that such a damped model is \textit{market consistent}. 
\bigskip

\section{Calibration Results} \label{section:calibration_results}
In order to demonstrate the applicability and effectiveness of the damping approach proposed in Section \ref{sec:technical_results}, we calibrate a market model to 1-year EURIBOR market data from May 15, 2023. We then extrapolate from the given market data to simulate term rates from the damped model with maturities of 60 years in the future. It should be noted that at the time of writing, the eurozone had not yet transitioned to short-term overnight interest rates in the form of the ESTR. However, as in-arrears 1-year EURIBOR caplet volatilities equal the corresponding ESTR caplet volatilities, the market model modelling the backward-looking term rates derived from ESTR is expected to suffer from similar blow-up issues in the current market environment. The market data used to calibrate the model can be found in Tables \ref{tab:1YearFwdCurve}, \ref{tab:15Y_caplet_vols} and \ref{tab:15Y_swaption_volas}. \\

\begin{table}[ht]
\centering
\resizebox{\columnwidth}{!}{
    \begin{tabular}{lllllllllll}
        {$T$}     & 0Y      & 1Y      & 2Y      & 3Y      & 4Y      & 5Y      & 6Y      & 7Y      & 8Y      & 9Y     \\ \hline
        {$\%$} & $3.795$ & $3.062$ & $2.712$ & $2.734$ & $2.898$ & $3.023$ & $3.111$ & $3.232$ & $3.337$ & $3.405$ \\
                                  &         &         &         &         &         &         &         &         &         &         \\
        {}     & 10Y     & 11Y     & 12Y     & 13Y     & 14Y     & 15Y     & 20Y     & 25Y     & 30Y     &         \\ \hline
        {$\%$} & $3.414$ & $3.426$ & $3.348$ & $3.291$ & $3.186$ & $2.986$ & $2.293$ & $1.930$ & $1.768$ &        
    \end{tabular}
}
\caption{1-year EURIBOR forward rates, i.e. $F(0,T,T+1)$, as of May 15, 2023. Source: Refinitiv Datastream.}
\label{tab:1YearFwdCurve}
\bigskip
\centering
\resizebox{0.8\columnwidth}{!}{
    \begin{tabular}{lllllllll}
             & 1Y      & 2Y      & 3Y      & 4Y      & 5Y      & 6Y      & 7Y      & 8Y      \\ \hline
        {$\%$} & $37.64$ & $44.56$ & $45.75$ & $40.05$ & $38.06$ & $34.73$ & $33.49$ & $29.64$ \\
                                  &         &         &         &         &         &         &         &         \\
        {}     & 9Y      & 10Y     & 11Y     & 12Y     & 13Y     & 14Y     & 15Y     &         \\ \hline
        {$\%$} & $28.77$ & $27.46$ & $26.14$ & $26.35$ & $26.86$ & $26.36$ & $27.67$ &        
    \end{tabular}
}
\caption{1Y-EURIBOR ATM caplet volatilities, $\left\{\sigma^{BS}_{i}\right\}_{i=2}^{16}$, as of May 15, 2023. The index of the x-axis represents the expiry tenor $T_{i-1}$. Source: Refinitiv Datastream.}
\label{tab:15Y_caplet_vols}
\bigskip
\centering
\resizebox{\columnwidth}{!}{
    \begin{tabular}{l|llllllllll}
    \%  & 1Y      & 2Y      & 3Y      & 4Y      & 5Y      & 6Y      & 7Y      & 8Y      & 9Y      & 10Y     \\ \hline
    1Y  & $38.15$ & $42.05$ & $41.38$ & $40.8$  & $40.05$ & $39.16$ & $38.19$ & $37.26$ & $36.32$ & $35.24$ \\
    2Y  & $44.61$ & $44.84$ & $42.78$ & $41.48$ & $40.23$ & $39.15$ & $38.06$ & $37.14$ & $36.2$  & $35.26$ \\
    3Y  & $43.29$ & $42.57$ & $40.87$ & $39.49$ & $38.16$ & $37.12$ & $36.15$ & $35.29$ & $34.46$ & $33.65$ \\
    4Y  & $40.15$ & $39.77$ & $38.26$ & $36.95$ & $35.72$ & $34.81$ & $34.05$ & $33.24$ & $32.57$ & $31.89$ \\
    5Y  & $37.81$ & $37.28$ & $35.91$ & $34.6$  & $33.4$  & $32.59$ & $31.81$ & $31.27$ & $30.8$  & $30.36$ \\
    7Y  & $33.33$ & $32.85$ & $31.41$ & $30.43$ & $29.78$ & $29.14$ & $28.65$ & $28.32$ &         &         \\
    10Y & $28.77$ & $28.57$ & $27.82$ & $27.44$ & $27.12$ &         &         &         &         &        
    \end{tabular}
}
\caption{1-year EURIBOR ATM swaption volatilities, $\left\{\sigma^{BS}_{i,j-i}\right\}_{i < j}$, as of May 15, 2023. The y-axis' index represents the expiry tenor $T_i$, while the x-axis' index denotes the period to maturity, $T_j - T_i$. Source: Refinitiv Datastream.}
\label{tab:15Y_swaption_volas}
\end{table}

In our model, the forward rates are driven by the instantaneous volatilities 
\[{\sigma_i(t) = g(T_{i-1}- t) \, f\left(\phi_{i}(t)\right) \, v_i\left(\phi_{i}(t)\right) \quad \, 2 \leq i \leq 60, \ 0 \leq t \leq T_{i-1}.}\] Specifically, we choose ${g(t) = (x_1 + x_2 t + x_3 t^2) \exp\left(- x_4 t\right) + x_5,}$ which is a classical Rebonato-type principal factor (cf. \cite[Sec. 21.3]{RebonatoRiccardo}). Moreover,
\[{\phi_i(t) = V\left(\int_0^t g_i(s)^2 \mathrm{d}s\right)},\] 
$V$ is the total implied variance structure, $f$ its associated damping factor and $\left\{v_i\right\}_{i=2}^{60}$ is a set of correlation factors. The damping factors are chosen as in Examples \ref{natural_damping} and \ref{vol_freeze} with ${\epsilon = 0.01}$ and ${\delta = 0}$. In particular, this means that there is only one damping threshold $\tau_1 = \tau$ in both cases. In this chapter, we refer to this $\tau$ as "the" damping threshold. To test the effectiveness of the decorrelation-beyond-a-threshold method, $\left\{v_i\right\}_{i=2}^{60}$ are either chosen as described in Example \ref{decorrelation} with threshold $\tau$ or in such a way that they preserve the correlation structure fitted to the market data.

\subsection{Calibration} \label{subsection:calibration}
Let $\{0,1,\dots,60\}$ be the underlying tenor structure of the market models and denote by $K=15$ the index of the largest tenor which is considered in the calibration procedure. The market model of the 1-year EURIBOR is calibrated to the market data in the usual three step process: \\

\begin{enumerate}
\item To begin with, we fit the available initial forward rates, $\left\{F_i(0)\right\}_{i}$, to a Nelson-Siegel type family by means of least squares (c.f. \cite[Ch. 3.3]{Filipovic}). \\

\item Secondly, the principle factor $g$ is fitted to the market data. Hereby, we assume that the damping threshold $\tau$ is larger or equal than the minimum permissible threshold $\tau \geq 0.9551$ in accordance with the Market Consistency Criterion (\ref{consistency_caps}). This already determines the lengths of the instantaneous volatilities. \\

\item In the final step, the instantaneous correlations are fitted to the swaption data given the caplet volatilities to obtain the directions of the instantaneous volatilities. To guarantee that the correlation matrix is positive definite with positive entries and to avoid numerical instabilities when fitting the correlations to the market data, the forward rates' instantaneous correlations, $\left\{\rho_{i,j}\right\}_{1 \leq i,j \leq N}$, are parametrised by 
\begin{align*}
\rho_{i,j} = \exp\Bigg(-&\frac{|j-i|}{N-1}\Bigg(-\log\left(\rho_\infty\right)\\
&+ \eta_1 \frac{i^2 + j^2 + i\,j - 3\,N\,i - 3\,N\,j +3\,i + 3\,j +2\,N^2 -N - 4}{(N-2)(N-3)} \\
&+ \eta_2 \frac{i^2 + j^2 + i\,j - N\,i - N\,j - 3\,i - 3\,j + 3\,N + 2}{(N-2)(N-3)}\Bigg)\Bigg)
\end{align*}
for every $1 \leq i,j \leq N$ with $N=60$, \[0 \leq \eta_2 \leq 3\,\eta_1\] and \[{0 \leq \eta_1 + \eta_2 \leq -\log\left(\rho_\infty\right)}\] as proposed in \cite[Sec. 1]{Schoenmakers}. For further information regarding suitable correlation structures, the interested reader is referred to \cite[Sec. 6.9]{BrigoMercurio}. Ideally, we would like to fit the correlation structure to the swaption prices that are implied by the market model which is already fitted to caplet data at this point. Unfortunately, this would require the evaluation of swaption prices by means of Monte-Carlo simulation in every iteration of the minimisation procedure. Therefore, we apply the common approximation 
\begin{equation} \label{fwd_approx_MSF}
\left|\sigma^{\text{Swap}}_{i,j}\right|^2 \approx \sum_{p,q=i+1}^j \frac{v_p(0)v_q(0)F_p(0)F_q(0)\rho_{p,q}}{\left(R^{\text{Swap}}_{i,j}(0)\right)^2}\frac{1}{T_i}\int_0^{T_i}\left|\sigma_p(s)\right|\left|\sigma_q(s)\right|\mathrm{d}s,
\end{equation}
for $1 \leq i < j \leq 60$ (cf. \cite[Chpt. 11.5]{Filipovic}) to estimate the parameters $\{\eta_1, \eta_2, \rho_\infty\}$. Since the correlation matrix is positive definite, we find $\left\{u_i\right\}_{i=1}^{60} \subset \mathbb{R}^{60}$ such that $u_i^T u_j = \rho_{i,j}$ for every $1 \leq i,j, \leq 60$. Thus, the Market Consistency Criterion (\ref{consistency_swaptions}) reads $v_i(x) = u_i$ for every $0 \leq x \leq 0.9551$ and for all ${1 \leq i \leq 60}$. 
\end{enumerate}
\bigskip

\subsubsection{Calibration Results}
Figure \ref{fig:caplet_vols} suggests that the quoted implied caplet volatilities are well-matched by the caplet volatilities of the market model. Indeed, \[RMSE^{caplet} = 0.03885.\] The fitted parameters can be found in Table \ref{tab:least_squares_fit_Rebonato}. \\

In contrast, the swaption volatilities, that the models are calibrated to in the final step of the calibration procedure, could not be fitted as well as the caplet volatilities. Here, $RMSE^\text{swaption}$ amounts to $0.9408$. This assessment is reinforced by Figure \ref{fig:swaption_vols}. The estimated parameters of the instantaneous correlation can be found in Table \ref{tab:least_squares_fit_corr_params}. Finally, it is observed that the estimated correlations are very high, even for large differences of the forward rates' expiries which encourages interest rate blow-ups. 

\begin{table}[h]
\centering
    \begin{minipage}[t]{0.45\textwidth}
        \begin{tabular}{l|l|l|l|l}
        $x_1$     & $x_2$ & $x_3$    & $x_4$   & $x_5$    \\ \hline
        $0.02411$ & $0.0$ & $1.6393$ & $1.531$ & $0.1642$
        \end{tabular}
        \caption{Parameters of the model's instantaneous volatilities that were fitted to the caplet volatilities with expiries up to 15 years in the future.}
        \label{tab:least_squares_fit_Rebonato}
    \end{minipage} \hspace{0.3cm}
    \begin{minipage}[t]{0.45\textwidth}
        \begin{tabular}{l|l|l}
        $\eta_1$ & $\eta_2$ & $\rho_\infty$ \\ \hline
        $0.0999$ & $0.0$    & $0.9001$     
        \end{tabular}
        \caption{Fitted model calibration parameters}
        \label{tab:least_squares_fit_corr_params}
    \end{minipage}
\end{table}

\subsection{Dampening Effects}
To test the feasibility of the damping approach that is proposed in this paper, $3000$ instances of $F_{60}(59) = R(59,60)$, the interest rate that is accrued in $[59,60]$, are simulated under the spot measure $\mathbb{Q}^*$. To this end, we apply the Euler-Mayurama scheme (cf. \cite[Sec. 11.6]{Filipovic}) to simulate the logarithmic term rates: For every $1 \leq i \leq 60$, given instantaneous volatility processes $\left\{\sigma_i\right\}_{i=1}^{60}$, Itô's formula (cf. \cite[Sec. 4.1]{Filipovic}) implies
\[\mathrm{d}\log\left(F_i(t)\right) = \sigma_i(t)^T \mathrm{d}W_t^* + \sum_{j=\eta(t)+1}^i \frac{\Delta_j F_j(t)}{1+ \Delta_j F_j(t)} \sigma_i(t)^T \sigma_j(t) \mathrm{d}t - \frac{1}{2}\left|\sigma_i(t)\right|^2 \mathrm{d}t\] with $\eta(t) = l$ iff $t \in (T_{l-1},T_l]$ for $1 \leq l \leq N$ under $\mathbb{Q}^*$. Thus, given an equidistant grid $0 = t_0 < t_1 < \dots < t_N$ with ${t_{k+1} - t_k = h > 0} \ \ {\forall \, 1 \leq k \leq N}$, we can simulate $F_i$ at $t_k$ under the spot measure by 
\begin{align*}
\log\left(F_i(t_k)\right) = \log\left(F_i(t_{k-1})\right) &+  \sqrt{h}\,\sigma_i(t_{k-1})^TZ_{k-1} \\
&+ \sum_{j=\eta(t_{k-1})+1}^i \frac{\Delta_j F_j(t_{k-1})}{1+ \Delta_j F_j(t_{k-1})} \sigma_i(t_{k-1})^T \sigma_j(t_{k-1}) h \\
&- \frac{1}{2}\left|\sigma_i(t_{k-1})\right|^2 h
\end{align*}
${\forall \, 1 \leq k \leq N}$ and ${\forall \, 1 \leq i \leq 60}$. Here, $\left(Z_k\right)_{k=0}^{N-1}$ denotes a sequence of iid, $60$-dimensional, uncorrelated Gaussian vectors with standard normally distributed marginals. For the purpose of this case study, $h = \frac{1}{10}$. \\

\subsubsection{Simulation Results}
Before running large-scale simulations, we apply the ideas presented in the Appendix to predict the impact of our choice of damping thresholds: We are going to simulate $n=3000$ iid copies of the term rates $F_{60}(59)$ and we would like to choose the damping threshold $\tau$ such that the maximum of the simulated rates exceeds some $r > 0$ only with probability less than $p > 0$. For the purpose of this test, set $p = 20\%$. From Subsection \ref{subsubsection:damping_thresholds} we recall the following:
\begin{align*}
    \tau \geq \tau_\text{min} \quad &\implies \quad \text{market consistency}, \\
    \tau < \tau_\text{max} \quad &\implies \quad \text{damping effect}.
\end{align*}
In this example, $\tau_\text{min} = 0.9551$ and $\tau_\text{max} = 2.1416$. If we opt for the minimum permissible threshold $\tau_\text{min} = 0.9551$, Example \ref{appendix:damping_thresholds} shows that, for any $t \in [0,59]$, none of the $3000$ simulated values of $F_{60}(t)$ under $\mathbb{Q}^{60}$ will exceed $r_\text{min} = 0.7198$, and $r_\text{min} = 1.47549$ will be exceeded with probability $80\%$ if we apply the volatility freeze and the damping approach proposed by Desmettre et al. respectively. Furthermore, choosing $\tau_\text{max} = 2.1416$ increases this limit on the maximum of the $3000$ simulated values of $F_{60}(t)$ under $\mathbb{Q}^{60}$ to $r_\text{max} = 2.5081$. Thus, we choose the minimum permissible threshold $\tau_\text{min}$ for the simulations under $\mathbb{Q}^*$. As mentioned in Subsection \ref{subsection:total_implied_volatility_structures}, these damping methods generally reduce the model's implied caplet volatilities. A comparison of these extrapolated caplet volatilties can be found in Figure \ref{fig:caplet_vols_damped}.\\

Table \ref{tab:sim_res} presents an overview of the relative number of extreme simulated values of $R(59,60)$ and Figure \ref{fig:tail_comparison} exhibits a comparison of the interest rate's empirical distribution's tails.
 
\begin{table}[ht]
\centering
\begin{tabular}{r|lll}
$R(59,60)$                                  & $\geq 0.2$ & $\geq 0.7$ & $\geq 10^3$ \\ \hline
Undamped:                         & $0.2630$         & $0.2013$         & $0.0830$          \\
Decorrelation:                    & $0.1667$         & $0.1063$         & $0.0187$          \\
Desmettre et al.:                 & $0.2293$         & $0.1630$         & $0.0487$          \\
Desmettre et al. + Decorrelation: & $0.1730$         & $0.0937$         & $0.0163$          \\
VolFreeze:                         & $0.1460$         & $0.0770$         & $0.0083$          \\
VolFreeze + Decorrelation:        & $0.1520$         & $0.0770$         & $0.0073$         
\end{tabular}
\caption{Relative number of simulated values of $R(59,60)$ with respect to $\mathbb{Q}^*$ that exceed selected thresholds.}
\label{tab:sim_res}
\end{table}

Hereby, we find that reducing the instantaneous correlations is among the most effective methods to reduce the number of blow-ups. Note that Desmettre et al. came to the same conclusion in their simulation study (cf. \cite[Sec. 4.2]{Desmettre}). What is more, the volatility freeze and the combined method of volatility freeze and decorrelation leads to the least number of interest rate blow-ups, as expected. \\

\begin{figure}[h]
\centering
    \begin{minipage}[t]{0.48\textwidth}
        \includegraphics[width=1\textwidth]{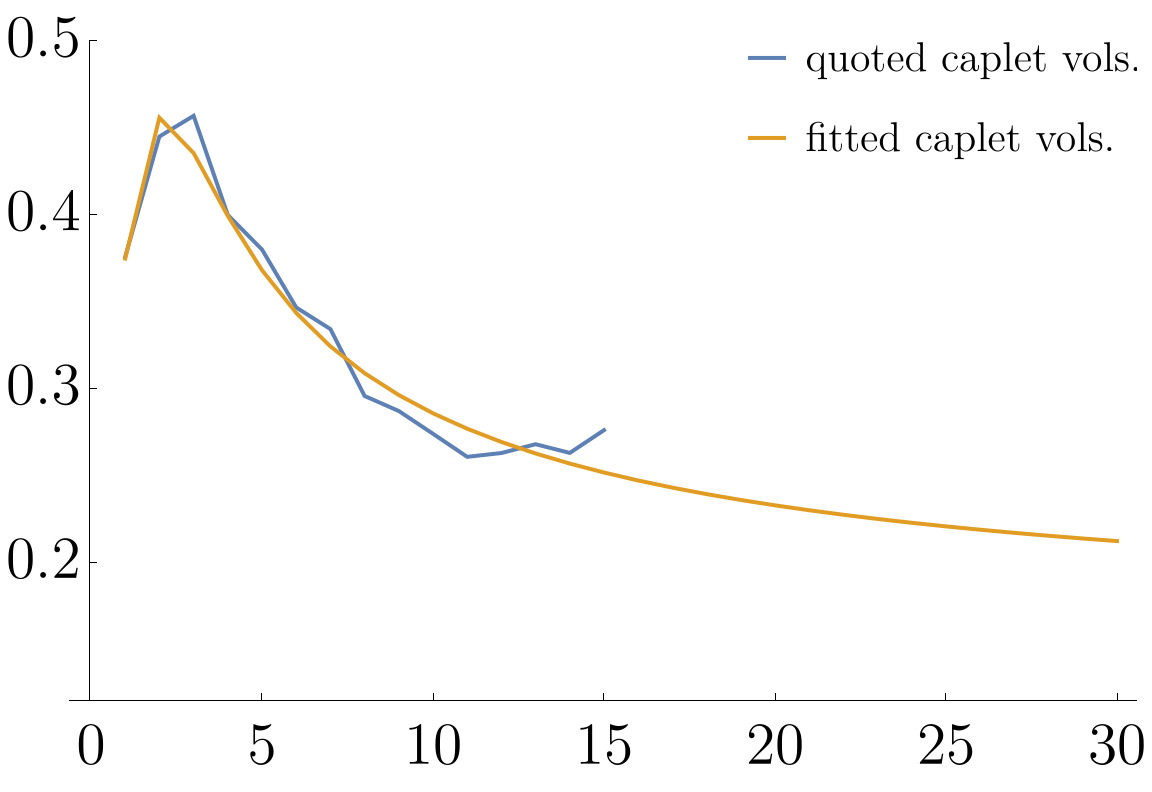}
        \caption{Model ATM caplet vol. versus quoted ATM caplet vol. with various expiries.}
        \label{fig:caplet_vols}
        \vspace{0.5cm}
    \end{minipage} \hspace{0.3cm}
    \begin{minipage}[t]{0.48\textwidth}
        \vspace{-5.53cm}
        \includegraphics[width=1\textwidth]{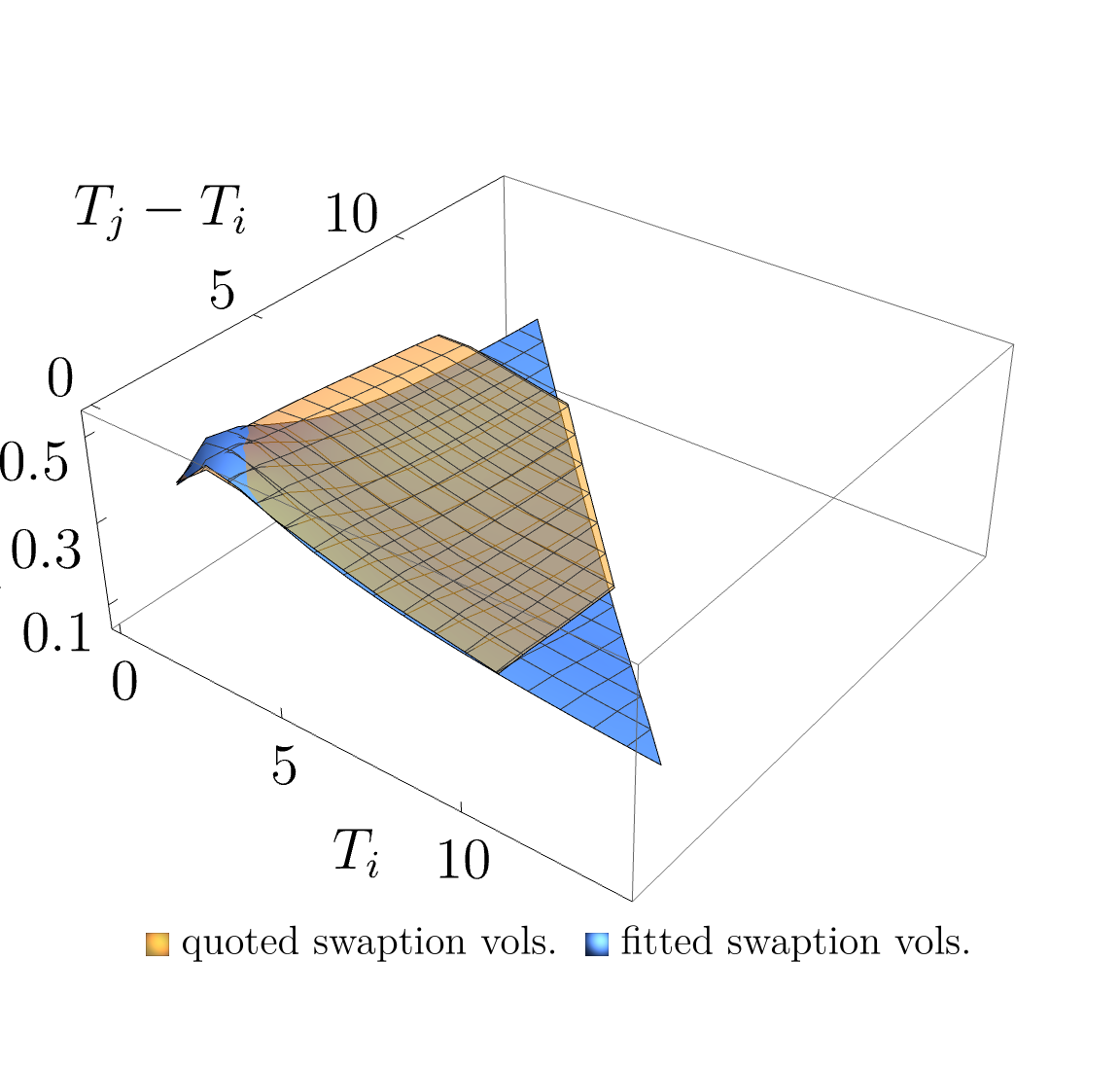}
        \vspace{-0.45cm}
        \caption{Model ATM swaption vol. versus quoted ATM swaption vol. with various expiries (x-axis) and periods to maturity (y-axis).}
        \label{fig:swaption_vols}
        \vspace{0.5cm}
    \end{minipage}
    \begin{minipage}[t]{0.48\textwidth}
        \includegraphics[width=1\textwidth]{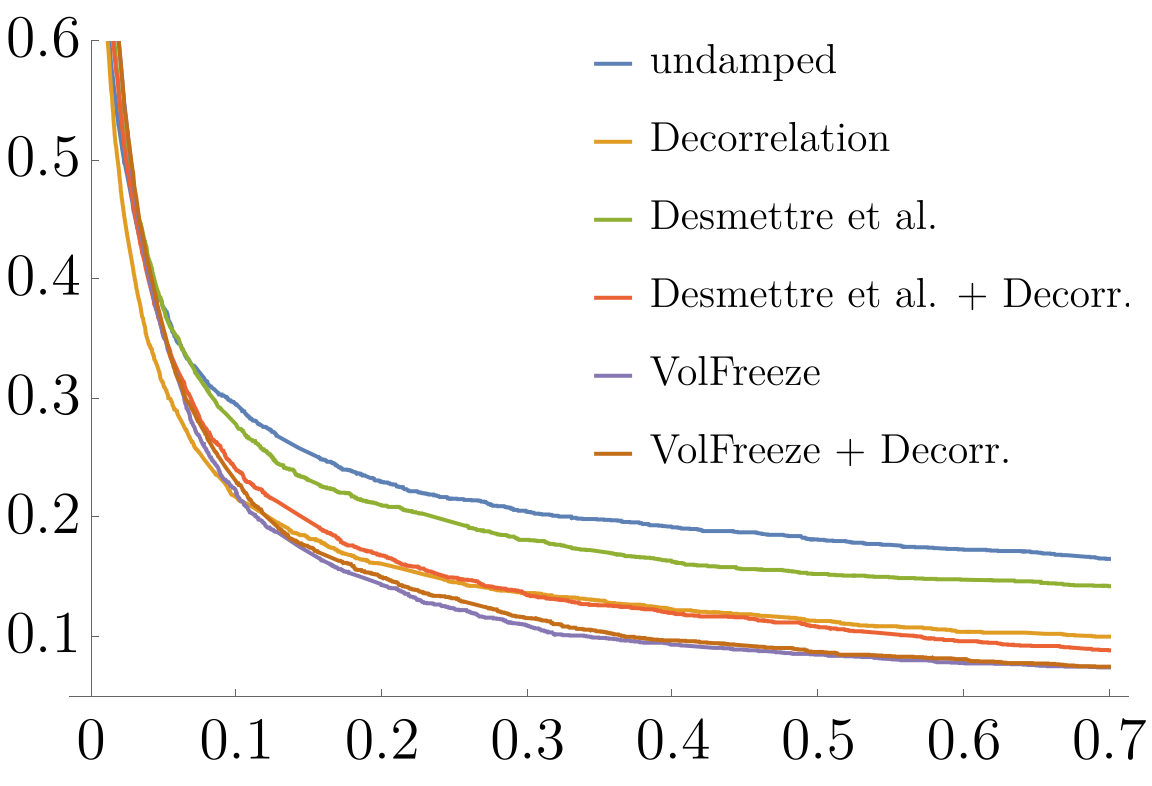}
        \caption{Empirical distribution's tails of $F_{60}(59)$ with different damping approaches.}
        \label{fig:tail_comparison}
    \end{minipage} \hspace{0.3cm}
    \begin{minipage}[t]{0.48\textwidth}
        \includegraphics[width=1\textwidth]{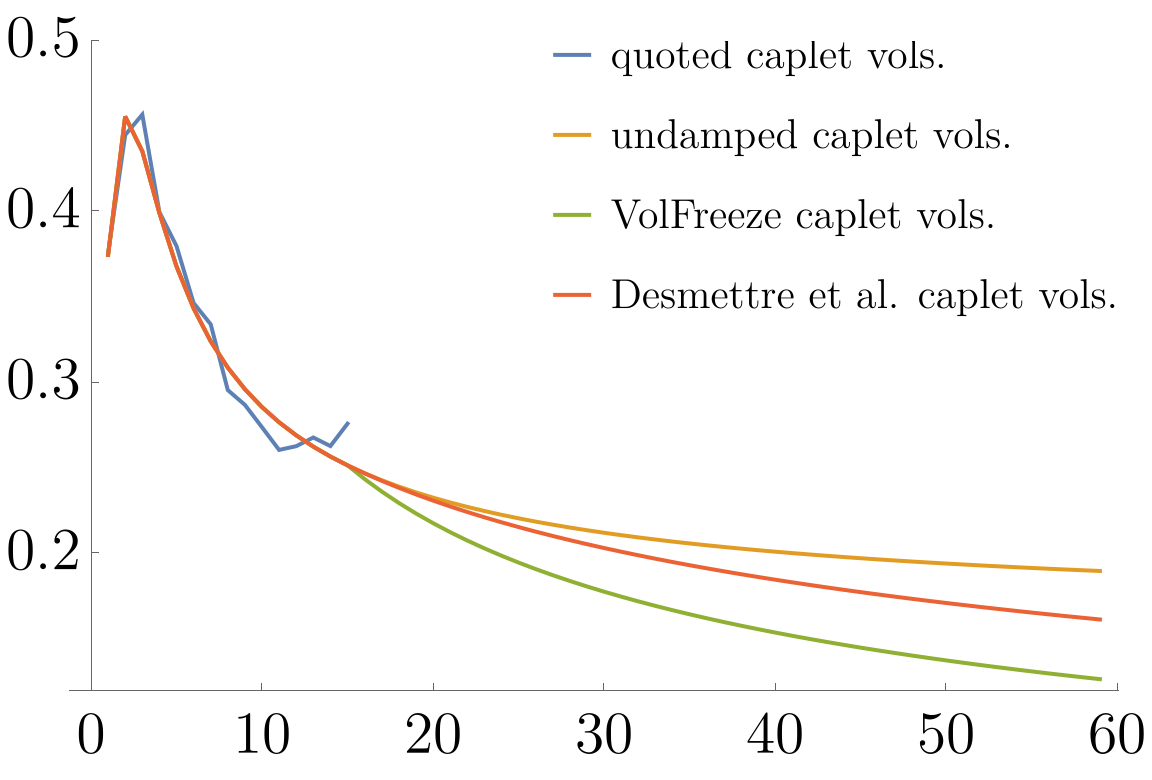}
        \caption{Extrapolated ATM caplet vol. with different damping approaches and various expiries.}
        \label{fig:caplet_vols_damped}
    \end{minipage}
\end{figure}

\bigskip

\section{Summary}
This contribution focuses on the mean-field extension of popular market models that are readily applied in the insurance industry. Since insurers are required to valuate their portfolios, some of which may be comprised of very long contracts, using arbitrage free interest rate scenarios, special emphasise is placed on the long-term extrapolation of term rate volatilities to obtain the best estimate of the market's expectation regarding future interest rate movements. \\

In particular, we extend the work of \cite{Desmettre} by developing a framework (see Section \ref{sec:technical_results}) that substantially reduces the computational effort required to simulate term rates from the given market models. Critically, this framework proves to be a minor complement to existing implementations of the Forward market model and can also be directly applied to in-arrears term rates stripped from short-term rates such as the ESTR. Furthermore, we strive to set out procedures and criteria to avoid arbitrariness with the theory developed in Subsection \ref{subsubsection:damping_thresholds}. Notably, this subsection presents criteria to avoid the radical damping or inflation of observed, valid caplet and swaption volatilities. \\

In Section \ref{section:calibration_results}, an exemplary simulation study was conducted based on the 1-year EURIBOR forward rates as of May 15, 2023 to test the feasibility of the proposed method. It is shown that the proposed damping approach considerably reduces the number of interest rate blow-ups while, crucially, preserving the term rates' martingale property. Furthermore, these results are consistent with the results of the simulation study by \cite{Desmettre}, which identify damping of the instantaneous correlations as a very effective method for reducing the number of interest rate blow-ups. In addition, the (pseudo) volatility freeze proposed in Example \ref{vol_freeze} also proves to be a very powerful, albeit radical, damping approach that could be used in combination with other, less radical damping approaches as a last resort.\\

\bigskip

\begin{appendices}
\section{Blow Up Probabilities} \label{section:appendix}
In this section, we propose a method that allows us to predict the effectiveness of a selection of damping thresholds in reducing the number of interest rate explosions before running large-scale simulations. \\

Suppose we would like to simulate $n$ iid copies $\left\{F_i^{(j)}(t)\right\}_{j=1}^n$ of the term rate $F_i(t)$ for large $i$ and $t$. Given a suitable set of segments, it might be natural to choose the damping thresholds $\left\{\tau_j^i\right\}_{j \in \mathbb{N}}$ in such a way that the maximum of the simulated term rates exceeds a threshold $r$ only with predefined probability less than $p > 0$. \\

Using the notation of Corollary \ref{bwd_corollary_special_existence}, recall that \[\log\left(F_i(t)\right) \sim \mathcal{N}\left(\log\left(F_i(0)\right) - \frac{1}{2}\int_0^t \left|\lambda_i(s)\right|^2 \, \mathrm{d}s, \ \int_0^t \left|\lambda_i(s)\right|^2 \, \mathrm{d}s\right)\] under $\mathbb{Q}^i$
and \[\phi_i(t) = \log\left(\frac{\mathbb{E}_{\mathbb{Q}^i}\left[F_i(t)^2\right]}{F_i(0)^2}\right) = \int_0^t \left|\lambda_i(s)\right|^2 \, \mathrm{d}s.\] For this reason, we have to ensure that
\begin{align*}
1 - p \leq \mathbb{Q}^i\left[\max_{1 \leq j \leq n} F_i^{(j)}(t) < r \right] = \Phi\left(\frac{\log\left(\frac{r}{F_i(0)}\right) + \frac{1}{2}\phi_i(t)}{\sqrt{\phi_i(t)}}\right)^n
\end{align*}
where $\Phi$ is the CDF of the standard normal distribution. This inequality holds iff
\begin{equation} \label{inequ:tau}
c_n(p) \sqrt{\phi_i(t)} \leq \log\left(\frac{r}{F_i(0)}\right) + \frac{1}{2}\phi_i(t)
\end{equation}
with $c_n(p) := \Phi^{(-1)}\left(\left(1-p\right)^\frac{1}{n}\right)$.
Note that 
\begin{equation} \label{appendix:root}
    x^*_i := 2\left(c_n(p)^2 -\log\left(\frac{r}{F_i(0)}\right) - \sqrt{c_n(p)^4-2 \log\left(\frac{r}{F_i(0)}\right) c_n(p)^2}\right)
\end{equation}
is the lesser of the two roots satisfying \[ c_n(p)^2 x = \left(\log\left(\frac{r}{F_i(0)}\right) + \frac{x}{2}\right)^2. \]
If we choose sufficiently small $0 < p < 1$ and sufficiently small $r > F_i(0)$, ${c_n(p)^2 > 2 \log\left(\frac{r}{F_i(0)}\right)}$. In this case, $x_i^* > 0 $ and (\ref{inequ:tau}) holds if 
\begin{equation} \label{damping_condition}
    \phi_i(t) \leq x^*_i \quad \Leftrightarrow \quad \int_{T_{i} - t}^{T_{i}} g(s)^2 \mathrm{d}s \leq V^{(-1)}( x^*_i ).
\end{equation} 
Clearly, not every choice of ${0 < p < 1}$, $r > 0$, damping thresholds and damping segments that satisfy the conditions of Subsection \ref{subsubsection:damping_thresholds} can fulfil this inequality. For this reason, it is advisable to first define the damping segments and the desired $0 < p < 1$ and only then to choose $r > 0$ and the damping thresholds according to Subsection \ref{subsubsection:damping_thresholds} and Equality (\ref{damping_condition}). Such an analysis is performed below for the damping factors introduced in Examples \ref{natural_damping} and \ref{vol_freeze}. \bigskip

\begin{example}[cont. of Examples \ref{natural_damping} and \ref{vol_freeze}] \label{appendix:damping_thresholds} ~\\
For arbitrary $\tau > 0$, fix $n >> 1$ and let \[ \left|\lambda_i(t)\right| = g\left(T_{i}-t\right) f\left(\phi_i(t)\right) \quad 1 \leq i \leq N\] where $f$ is defined as in Example \ref{natural_damping} or Example \ref{vol_freeze} with $\delta = 0$. In particular, this means that there is only one damping threshold $\tau_1 = \tau$ in both cases.
Moreover, let $K$ be the maximum tenor up to which market data is available. Here, we would even like to choose $\tau > 0$ such that ${1 - p \leq \mathbb{Q}^i\left[\max_{1 \leq j \leq n} F_i^{(j)}(t) < r \right]}$ for all $t \in [0,T_i]$ and $K < i \leq N$.
To ensure market consistency, we require \[ \tau \geq \tau_\text{min} = \max_{1 \leq i \leq K} \int_0^{T_{i}} g_i(s)^2 \mathrm{d}s = \int_0^{T_{K}} g(s)^2 \mathrm{d}s\] to fulfil the sufficient condition of Subsection \ref{subsubsection:damping_thresholds}. In order to observe a damping effect, recall that $\tau >  0$ must satisfy \[{\tau < \tau_\text{max} = \int_0^{T_{N}} g(s)^2 \mathrm{d}s}.\] 
Let $K < i \leq N$ and recall $x_i^*$ from Equation (\ref{appendix:root}). \\
In Example \ref{natural_damping}, Condition (\ref{damping_condition}) reads \[ \int_0^{T_{i}} g(s)^2 \mathrm{d}s \leq V^{(-1)}(x_i^*) = 
\begin{cases}
x_i^* & \text{if }  x_i^* \leq \tau, \\
\frac{\tau}{2} + \frac{\tau}{2} \, \exp\left(\frac{2 (x_i^*- \tau)}{\tau}\right) & \text{if } x_i^* > \tau.
\end{cases}\]
Generally, we therefore require $\tau < x_i^*$ which reduces to 
\begin{equation} \label{equation_tau_x_desmettre}
\int_0^{T_N} g(s)^2 \mathrm{d}s = \frac{\tau}{2} + \frac{\tau}{2} \, \exp\left(\frac{2 (x_i^*- \tau)}{\tau}\right)
\end{equation} 
if we would like to choose the same damping threshold $\tau$ for all $K < i \leq N$.
In Example \ref{vol_freeze}, Condition (\ref{damping_condition}) likewise reads \[ \int_0^{T_i} g(s)^2 \mathrm{d}s \leq V^{(-1)}(x_i^*) = \begin{cases}
x_i^* & \text{if } x_i^* \leq \tau, \\
\frac{x_i^* - (1 - \epsilon^2) \tau}{\epsilon^2} & \text{if } \tau < x_i^*.
\end{cases}\]
So, we require $\tau < x_i^*$ which again reduces to 
\begin{equation} \label{equation_tau_x_vol_freeze}
\int_0^{T_N} g(s)^2 \mathrm{d}s = \frac{x_i^* - (1 - \epsilon^2) \tau}{\epsilon^2}.
\end{equation} 

Given these constraints and some ${0 < p << 1}$ the question arises as to which $\tau$ and $r$ one might choose. To this end, note that the right hand sides of (\ref{equation_tau_x_desmettre}) and (\ref{equation_tau_x_vol_freeze}) decrease monotonically in $\tau$ for fixed $r$ and monotonically increase in $r$ for fixed $\tau$. These observations agree with our intuition, because the smaller $\tau$, the larger the damping effect and the smaller the (minimal) threshold rate $r$. If we denote the minimally permissible and maximally meaningful threshold rates $r$ by $r_{i, \text{min}}$ and $r_{i, \text{max}}$ respectively and interpret $x_i^*$ as a function of $r$, we obtain 
\begin{align*}
   x_i^*\left(r_{i,\text{min}}\right) &= \int_0^{T_{K}} g(s)^2 \mathrm{d}s + \frac{1}{2}\left(\int_0^{T_{K}} g(s)^2 \mathrm{d}s\right) \log\left(2\left(\frac{ \int_0^{T_i} g(s)^2 \mathrm{d}s}{\int_0^{T_{K}} g(s)^2 \mathrm{d}s}\right) -1\right), \\
   x_i^*\left(r_{i,\text{max}}\right) &= \int_0^{T_N} g(s)^2 \mathrm{d}s
\end{align*}
in Example \ref{natural_damping} and
\begin{align*}
   x_i^*\left(r_{i,\text{min}}\right) &= \epsilon^2 \int_0^{T_i} g(s)^2 \mathrm{d}s + (1 - \epsilon^2) \int_0^{T_{K}} g(s)^2 \mathrm{d}s, \\
   x_i^*\left(r_{i,\text{max}}\right) &= \int_0^{T_N} g(s)^2 \mathrm{d}s
\end{align*}
in Example \ref{vol_freeze}.
Denoting $x^*_{i,\min} := x_i^*\left(r_{i,\text{min}}\right)$ and ${x^*_{i,\max} := x_i^*\left(r_{i,\text{max}}\right)}$ and solving these equations for $r_{i,\text{min}}$ and $r_{i,\text{max}}$ yields
\[ r_{i,\text{min}} = \exp\left(c_n(p) \sqrt{x^*_{i,\min}} - \frac{x^*_{i,\min}}{2} \right) F_i(0),\]
\[ r_{i,\text{max}} =\exp\left(c_n(p) \sqrt{x^*_{i,\max}} - \frac{x^*_{i,\max}}{2} \right) F_i(0)\]
assuming $c_n(p)^2 \geq x_{i, \max}^*$. Thus, the globally minimally permissible and maximally meaningful threshold rates are $r_{\text{min}} = \underset{K < i \leq N}{\max} r_{i,\text{max}}$ respectively $r_{\text{max}} = \underset{K < i \leq N}{\min} r_{i,\text{max}}$. \\

The inequality $c_n(p)^2 \geq \underset{K < i \leq N}{\max} x_{i, \max}^*$ defines an upper bound for $p$ which is usually almost $1$, even for moderately small $n \in \mathbb{N}$. Furthermore, given such $p$, for every ${r \in \left[r_{\text{min}},r_{\text{max}}\right)}$ $x^*(r)$ is well-defined and there exists an unique $\tau(r) > 0$ such that Equations (\ref{equation_tau_x_desmettre}) and (\ref{equation_tau_x_vol_freeze}) are satisfied for all $K < i \leq N$. In the examples given, $\tau(r)$ can easily be calculated analytically.
\end{example} \bigskip
\end{appendices}

\bibliography{references-biblatex}


\begin{thebibliography}{12}
\ifx \bisbn   \undefined \def \bisbn  #1{ISBN #1}\fi
\ifx \binits  \undefined \def \binits#1{#1}\fi
\ifx \bauthor  \undefined \def \bauthor#1{#1}\fi
\ifx \batitle  \undefined \def \batitle#1{#1}\fi
\ifx \bjtitle  \undefined \def \bjtitle#1{#1}\fi
\ifx \bvolume  \undefined \def \bvolume#1{\textbf{#1}}\fi
\ifx \byear  \undefined \def \byear#1{#1}\fi
\ifx \bissue  \undefined \def \bissue#1{#1}\fi
\ifx \bfpage  \undefined \def \bfpage#1{#1}\fi
\ifx \blpage  \undefined \def \blpage #1{#1}\fi
\ifx \burl  \undefined \def \burl#1{\textsf{#1}}\fi
\ifx \doiurl  \undefined \def \doiurl#1{\url{https://doi.org/#1}}\fi
\ifx \betal  \undefined \def \betal{\textit{et al.}}\fi
\ifx \binstitute  \undefined \def \binstitute#1{#1}\fi
\ifx \binstitutionaled  \undefined \def \binstitutionaled#1{#1}\fi
\ifx \bctitle  \undefined \def \bctitle#1{#1}\fi
\ifx \beditor  \undefined \def \beditor#1{#1}\fi
\ifx \bpublisher  \undefined \def \bpublisher#1{#1}\fi
\ifx \bbtitle  \undefined \def \bbtitle#1{#1}\fi
\ifx \bedition  \undefined \def \bedition#1{#1}\fi
\ifx \bseriesno  \undefined \def \bseriesno#1{#1}\fi
\ifx \blocation  \undefined \def \blocation#1{#1}\fi
\ifx \bsertitle  \undefined \def \bsertitle#1{#1}\fi
\ifx \bsnm \undefined \def \bsnm#1{#1}\fi
\ifx \bsuffix \undefined \def \bsuffix#1{#1}\fi
\ifx \bparticle \undefined \def \bparticle#1{#1}\fi
\ifx \barticle \undefined \def \barticle#1{#1}\fi
\bibcommenthead
\ifx \bconfdate \undefined \def \bconfdate #1{#1}\fi
\ifx \botherref \undefined \def \botherref #1{#1}\fi
\ifx \url \undefined \def \url#1{\textsf{#1}}\fi
\ifx \bchapter \undefined \def \bchapter#1{#1}\fi
\ifx \bbook \undefined \def \bbook#1{#1}\fi
\ifx \bcomment \undefined \def \bcomment#1{#1}\fi
\ifx \oauthor \undefined \def \oauthor#1{#1}\fi
\ifx \citeauthoryear \undefined \def \citeauthoryear#1{#1}\fi
\ifx \endbibitem  \undefined \def \endbibitem {}\fi
\ifx \bconflocation  \undefined \def \bconflocation#1{#1}\fi
\ifx \arxivurl  \undefined \def \arxivurl#1{\textsf{#1}}\fi
\csname PreBibitemsHook\endcsname

%%% 1
\bibitem[\protect\citeauthoryear{Desmettre et~al.}{2022}]{Desmettre}
\begin{botherref}
\oauthor{\bsnm{Desmettre}, \binits{S.}},
\oauthor{\bsnm{Hochgerner}, \binits{S.}},
\oauthor{\bsnm{Omerovic}, \binits{S.}},
\oauthor{\bsnm{Thonhauser}, \binits{S.}}:
A mean-field extension of the LIBOR market model.
International Journal of Theoretical and Applied Finance
(2022).
\doiurl{10.1142/S0219024922500054}
\end{botherref}
\endbibitem

%%% 2
\bibitem[\protect\citeauthoryear{of~the European~Union}{2015}]{2015Regulation}
\begin{botherref}
\oauthor{\bsnm{European~Union}, \binits{T.C.}}:
COMMISSION DELEGATED REGULATION (EU) 2015/35 of 10 October 2014 supplementing
  Directive 2009/138/EC of the European Parliament and of the Council on the
  taking-up and pursuit of the business of Insurance and Reinsurance (Solvency
  II).
Official Journal of the European Union
(2015).
\url{http://data.europa.eu/eli/reg_del/2015/35/oj}
\end{botherref}
\endbibitem

%%% 3
\bibitem[\protect\citeauthoryear{Gach and Hochgerner}{2022}]{GachHochgerner}
\begin{barticle}
\bauthor{\bsnm{Gach}, \binits{F.}},
\bauthor{\bsnm{Hochgerner}, \binits{S.}}:
\batitle{Estimation of future discretionary benefits in traditional life
  insurance}.
\bjtitle{ASTIN Bulletin: The Journal of the IAA}
\bvolume{52},
\bfpage{835}--\blpage{876}
(\byear{2022})
\doiurl{10.1017/asb.2022.16}
\end{barticle}
\endbibitem

%%% 4
\bibitem[\protect\citeauthoryear{Hochgerner and Gach}{2019}]{GachHochgerner2}
\begin{barticle}
\bauthor{\bsnm{Hochgerner}, \binits{S.}},
\bauthor{\bsnm{Gach}, \binits{F.}}:
\batitle{Analytical validation formulas for best estimate calculation in
  traditional life insurance}.
\bjtitle{European Actuarial Journal}
\bvolume{9},
\bfpage{423}--\blpage{443}
(\byear{2019})
\doiurl{10.1007/s13385-019-00212-2}
\end{barticle}
\endbibitem

%%% 5
\bibitem[\protect\citeauthoryear{Vedani et~al.}{2017}]{Verdani}
\begin{botherref}
\oauthor{\bsnm{Vedani}, \binits{J.}},
\oauthor{\bsnm{El~Karoui}, \binits{N.}},
\oauthor{\bsnm{Loisel}, \binits{S.}},
\oauthor{\bsnm{Prigent}, \binits{J.-L.}}:
Market inconsistencies of {\it market-consistent} {E}uropean life insurance
  economic valuations: pitfalls and practical solutions.
European Actuarial Journal
\textbf{7}
(2017)
\doiurl{10.1007/s13385-016-0141-z}
\end{botherref}
\endbibitem

%%% 6
\bibitem[\protect\citeauthoryear{Lyashenko and
  Mercurio}{2019a}]{MercurioLyashenko}
\begin{botherref}
\oauthor{\bsnm{Lyashenko}, \binits{A.}},
\oauthor{\bsnm{Mercurio}, \binits{F.}}:
Looking forward to backward-looking rates: A modeling framework for term rates
  replacing libor
(2019)
\doiurl{10.2139/ssrn.3330240}
\end{botherref}
\endbibitem

%%% 7
\bibitem[\protect\citeauthoryear{Lyashenko and
  Mercurio}{2019b}]{MercurioLyashenko2}
\begin{botherref}
\oauthor{\bsnm{Lyashenko}, \binits{A.}},
\oauthor{\bsnm{Mercurio}, \binits{F.}}:
Looking forward to backward-looking rates: Completing the generalized forward
  market model
(2019)
\end{botherref}
\endbibitem

%%% 8
\bibitem[\protect\citeauthoryear{Filipovic}{2009}]{Filipovic}
\begin{bbook}
\bauthor{\bsnm{Filipovic}, \binits{D.}}:
\bbtitle{Term-Structure Models. A Graduate Course.}
\bpublisher{Springer},
\blocation{Heidelberg}
(\byear{2009}).
\doiurl{10.1007/978-3-540-68015-4}
\end{bbook}
\endbibitem

%%% 9
\bibitem[\protect\citeauthoryear{Musiela and
  Rutkowski}{2006}]{MusielaRutkowski}
\begin{bbook}
\bauthor{\bsnm{Musiela}, \binits{M.}},
\bauthor{\bsnm{Rutkowski}, \binits{M.}}:
\bbtitle{Martingale Methods in Financial Modelling}
vol. \bseriesno{36}.
\bpublisher{Springer},
\blocation{Heidelberg}
(\byear{2006}).
\doiurl{10.1007/b137866}
\end{bbook}
\endbibitem

%%% 10
\bibitem[\protect\citeauthoryear{Rebonato}{2005}]{RebonatoRiccardo}
\begin{bbook}
\bauthor{\bsnm{Rebonato}, \binits{R.}}:
\bbtitle{Volatility and Correlation: the Perfect Hedger and the Fox}.
\bpublisher{John Wiley \& Sons},
\blocation{Chichester}
(\byear{2005}).
\doiurl{10.1002/9781118673539}
\end{bbook}
\endbibitem

%%% 11
\bibitem[\protect\citeauthoryear{Schoenmakers}{2002}]{Schoenmakers}
\begin{barticle}
\bauthor{\bsnm{Schoenmakers}, \binits{J.}}:
\batitle{Calibration of libor models to caps and swaptions: a way around
  intrinsic instabilities via parsimonious structures and a collateral market
  criterion}.
\bjtitle{Preprint, Weierstraß-Institut für Angewandte Analysis und
  Stochastik}
(\byear{2002})
\doiurl{10.20347/WIAS.PREPRINT.740}
\end{barticle}
\endbibitem

%%% 12
\bibitem[\protect\citeauthoryear{Brigo and Mercurio}{2006}]{BrigoMercurio}
\begin{bbook}
\bauthor{\bsnm{Brigo}, \binits{D.}},
\bauthor{\bsnm{Mercurio}, \binits{F.}}:
\bbtitle{Interest Rate Models-theory and Practice: with Smile, Inflation and
  Credit}
vol. \bseriesno{2}.
\bpublisher{Springer},
\blocation{Heidelberg}
(\byear{2006}).
\doiurl{10.1007/978-3-540-34604-3}
\end{bbook}
\endbibitem

\end{thebibliography}

\end{document}